\documentclass[journal]{IEEEtran}
\usepackage[cmex10]{amsmath}
\usepackage{cite}
\usepackage{amsthm}
\usepackage{graphicx}
\usepackage{bm,mathrsfs}
\usepackage{dsfont}
\usepackage{array}
\usepackage{mdwmath}
\usepackage{mdwtab}
\usepackage{fixltx2e}
\usepackage{url}

\theoremstyle{definition}

\newtheorem{thm}{Theorem}
\newtheorem{lem}[thm]{Lemma}

\newcommand{\ket}[1]{\lvert #1 \rangle}

\begin{document}

\title{Reducing the quantum communication cost of quantum secret sharing}
\author{Ben Fortescue, Gilad Gour
\thanks{B. Fortescue is with the Department of Physics, Southern Illinois University, 1245 Lincoln Dr, Carbondale 62901, USA.  Most work in this paper done while with the Institute for Quantum Information Science, University of Calgary, 2500 University Dr NW, Calgary, Alberta, Canada.}
\thanks{G. Gour is with the Institute for Quantum Information Science and the Department of
Mathematics and Statistics, University of Calgary, 2500 University Dr NW, Calgary, Alberta, Canada.}
\thanks{B. Fortescue was funded by AITF, NSERC Project ``Frequency'' and a PIMS postdoctoral fellowship.  G. Gour is funded by NSERC.}}
\date{\today}
\maketitle

\begin{abstract}
We demonstrate a new construction for perfect quantum secret sharing (QSS) schemes based on imperfect ``ramp'' secret sharing
combined with classical encryption, in which the individual parties' shares are split into quantum and classical
components, allowing the former to be of lower dimension than the secret itself.   We show that such
schemes can be performed with smaller quantum components and lower overall quantum communication than required for existing methods.
We further demonstrate that one may combine both imperfect quantum and imperfect classical secret sharing to produce an overall
perfect QSS scheme, and that examples of such schemes (which we construct) can have the smallest
quantum and classical share components possible for their access structures, something provably not achievable using perfect underlying schemes.
Our construction has significant potential for being adapted to other QSS schemes based on stabiliser codes.
\end{abstract}

\section{Introduction}
Quantum secret sharing (QSS) is a cryptographic protocol in which a dealer
encodes a quantum state (the secret) into multiple shares and distributes the shares to various
players.  Certain subsets of the players (denoted ``authorised'' sets, and collectively
known as the access structure) can collaboratively reconstruct the
secret from their shares while other subsets (denoted ``forbidden'' sets, and collectively known
as the adversary structure) can obtain no information about the secret from
their shares.  An example would be an $n$-player protocol where the access structure
consists of all sets of $k$ or more players and the adversary structure of all sets of $k-1$ or fewer players.
Such protocols are known as $(k,n)$ threshold schemes.
Analogous classical secret sharing (CSS) protocols (involving classical secrets and shares) were introduced by
Shamir \cite{Shamir79} and independently by Blakley \cite{Blakley79}.  QSS protocols were originally described by Hillery, Bu{\^z}ek and Berthiaume \cite{HBB99}
and Cleve, Gottesman and Lo \cite{CGL}, with the latter giving a means to construct a threshold scheme
for any allowable \footnote{QSS protocols must satisfy the two criteria of monotonicity (an authorised
subset remains authorised if shares are added) and the no-cloning theorem (two disjoint subsets cannot both be authorised,
since then the players could separately reconstruct two copies of the secret, violating no-cloning)} access structure.

Threshold schemes are examples of ``perfect'' access structures, i.e.\ those in which every subset of players
is either authorised or forbidden, with no subsets able to reconstruct only partial information
about the secret.  As proven by Gottesman \cite{Gott00}, perfect QSS requires that all player shares (aside
from trivial shares whose presence never affects whether or not a subset is authorised) be
at least as large as the secret.   That is, for a secret of dimension $d_s$, and players $i$ each
receiving a share of dimension $d_i$, perfect QSS requires $\min_i d_i\ge d_s$, with ``optimal''
schemes achieving $d_i=d_s$ for all $i$.

This bound imposes a large potential cost on the communication and storage of shares in perfect QSS.  For example,
sharing a 1-qubit secret by distributing quantum shares to 100 players in a threshold scheme will require at least
100 qubits to be communicated by the dealer and (due to the no-cloning bound) at least 51 qubits
to be used by the players for reconstruction.  This may well not be practical given noisy quantum communication
channels, unreliable joint quantum operations, and short storage times.  Furthermore, no optimal general construction
for schemes with non-threshold perfect access structures is known
(they may be non-optimally constructed by concatenating threshold schemes \cite{Gott00}), so the costs for such schemes may be larger still.

In this paper we demonstrate a new class of protocols which reduce the cost of quantum communication
and storage in implementing perfect QSS schemes, by combining two existing ideas for doing so: ramp (i.e. non-perfect) secret
sharing \cite{OSIY05} and hybrid secret sharing \cite{NMI01,SS05}.  Our protocols, involving both classical and quantum shares,
require a smaller amount of total quantum
communication than either approach taken in isolation. Moreover, we describe protocols
that are provably optimal in minimising both the size of the quantum shares and,
given this minimisation, the size of the classical shares required.

\section{Ramp secret sharing}
The derivation of the share size bound in \cite{Gott00} (which we discuss in more detail
in Section \ref{sec-bound}) depends on the observation that, for a {\it perfect} QSS, there must exist some
forbidden subset which can be made authorised by adding a single additional share.
Thus complete information about the secret of size $d_s$ may be transferred
via a share of size $d_i$, from which the bound follows.

Consider now an access structure in which every forbidden subset instead requires a minimum of $L$ shares
to be added to become authorised: the corresponding bound would be (in the simplified
case where all shares are of the same dimension $d_i$) that $L\log_2 d_i\ge \log_2 d_s$.
There would also now exist certain {\it intermediate} subsets (consisting of some
forbidden subset with $l<L$ shares added) who could reconstruct some {\it partial}
information about the secret.  Thus, one could potentially have smaller shares
(encoding e.g. an $s$-qudit secret into shares of size $s/L$ qudits), but at the cost of
some security; there will be some information leakage to subsets who are not authorised.

Such schemes are known as ``ramp'' secret sharing and were originally proposed
for CSS \cite{BM84, Yam86} for which an analogous share size bound exists (and can lead to very large
data storage requirements for perfect schemes), and one can formalise
whether or not such leakage is tolerable by considering a computationally-bounded
adversary.  Ramp QSS (RQSS) schemes have received relatively little
attention in the literature, but a construction has been given by Ogawa et al. \cite{OSIY05}
for all allowable $(k,L,n)$ threshold access structures i.e.\ where subsets of $k$ or more players are authorised,
those of $k-L$ or fewer are forbidden, and those of $k-l$ ($l<L$) are intermediate (so a perfect threshold scheme in this notation
has $L=1$).  Furthermore,
the construction of \cite{OSIY05} is optimal in that it encodes
a secret of $L$ qudits into shares of a single qudit (though we note that, like the protocol
of \cite{CGL} upon which it is based, this protocol additionally requires
the qudit shares to be of prime dimension $d\ge n$, and hence not all ramp threshold schemes are covered
by this construction).

\section{Hybrid secret sharing}
In addition to being applicable to perfect schemes only, the share size bound for QSS
does not require that every share be a quantum system i.e.\ some shares can potentially consist
partly or solely of {\it classical} information.  While strictly speaking such shares can still be regarded
as quantum systems and must still obey the size bound, they would clearly have great practical advantages,
since classical information is far easier to communicate, store and process, and they would allow QSS
schemes to involve players with no ability to handle quantum states.  Such schemes can be implemented
through combining QSS with classical encryption.

In the well-known quantum teleportation protocol \cite{teleport}, Alice
transmits an unknown quantum state $\psi$ to Bob (with whom she
shares a maximally-entangled state) via a local measurement on her joint quantum
system followed by the transmission of some classical information $C$ to Bob (who then performs a state
reconstruction conditioned on $C$).  Bob's local quantum system can therefore, immediately prior
to receiving the information,
be represented as being in some mixed state $\rho_B$ which must be independent of $\psi$ due to the no-signalling
theorem, but from which, in combination with $C$, $\psi$ can be reliably obtained.  In other
words, one can consider $\rho_B$ as being the state $\psi$ securely encrypted using a classical key $C$.
Specifically, it follows from the teleportation protocol that
one can securely encrypt a quantum state of dimension $d$ using a classical key of 2$\log_2 d$
bits (which has also been shown \cite{AMTW00} to be the minimum key size required), via the basis-state encryption:
\begin{equation}\label{eq-encrypt}
\ket{j}\to \omega^{jl}\ket{j+k\textrm{ mod }d}
\end{equation}
where $\omega=e^{\frac{2i\pi}{d}}$ and the two values $\{k,l\}\in \{0\ldots d\}$ (chosen at random and thus independent
of the original state) constitute the classical key (CK)
for the encrypted quantum secret (EQS).

Such encryption can be exploited in QSS, as demonstrated in \cite{HBB99}, which effectively
described (in terms of teleportation) a scheme in which the dealer sends one player the EQS
and the other the CK.  Neither player alone has any information about the secret but together
they may reconstruct it, thus constituting a perfect $(2,2)$ threshold scheme in which one
share is the size of the secret but the other is wholly classical.  Though not discussed in terms of communication
cost in \cite{HBB99}, this halves the amount of quantum communication required vs. a standard optimal QSS.

This idea was generalised by Nascimento, Mueller-Quade and Imai \cite{NMI01}, who considered
the case of such ``hybrid'' QSS (HQSS) schemes where the EQS and CK are separately encoded using QSS and CSS schemes
respectively, and then the classical and quantum shares distributed to the players.  They gave a general
construction for hybrid schemes; we rephrase its specific application to the case of threshold schemes
in the following lemma.

\begin{lem}\label{lem-ht}\cite{NMI01}
A hybrid $(k,n)$ threshold QSS can be constructed using only $(2n-2k+1)$ fully- or partly-quantum shares
with the remaining shares fully classical.
\end{lem}
\begin{IEEEproof}
The construction consists of the dealer classically encrypting the secret, dividing the EQS
among a subset of $n_q\le n$ players using a $(k_q, n_q)$ QSS, and the CK among all $n$ players
using a $(k,n)$ CSS (so some players have both classical and quantum shares).
Thus $k-1$ or fewer players will have no information about the CK
and hence none about the secret (irrespective of any quantum shares they may also have).
$k$ or more players will be able to reconstruct the CK but we further require that they be able
to reconstruct the EQS, and hence have $k_q$ or more quantum shares.  $k$ players will
have a minimum of $k-(n-n_q)$ quantum shares, hence we require
\begin{equation}
n_q-k_q\ge n-k.\label{eq-kngap}
\end{equation}
To satisfy no-cloning, we additionally require
\begin{equation}
n_q\le 2k_q-1 \label{eq-noclone}.
\end{equation}
From (\ref{eq-kngap}) and (\ref{eq-noclone}) we obtain
\begin{equation}
n_q\ge 2n-2k+1
\end{equation}
and minimising $n_q$ by setting $n_q=2n-2k+1$ we achieve the result.
Note that due to no-cloning $n\ge (2n-2k+1)$ and hence in general this construction reduces the number of quantum shares required vs.
a non-hybrid scheme.
\end{IEEEproof}

Some observations about this protocol: we note that for $n_q=2n-2k+1$, it follows from (\ref{eq-kngap}) and (\ref{eq-noclone}) that
\begin{align}
n_q&=2k_q-1\textrm{ and}\\
n_q-k_q&=n-k.\label{eq-knequal}
\end{align}
for which there is only one solution for given $(k,n)$.
That is, in this construction, the optimal underlying QSS scheme is that with the unique access structure that both saturates the no-cloning bound
(we will refer to such schemes as ``boundary'' schemes) and satisfies (\ref{eq-knequal}).  For an initial $(k,n)$ boundary
scheme, then, the optimal underlying QSS scheme will simply have a $(k_q=k,n_q=n)$ structure.  Hence the construction of Lemma \ref{lem-ht}
can only reduce the
number of quantum shares for non-boundary threshold schemes.  An example of such an HQSS scheme is shown in Figure \ref{fig-hybrid}.

\begin{figure}
\includegraphics[height=6cm]{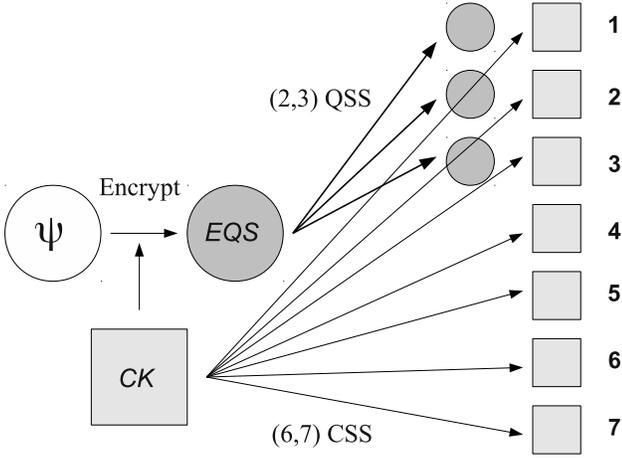}
\caption{An example (6,7) HQSS protocol.  A quantum secret $\psi$ is encrypted to an EQS using a CK, then the EQS is split among the first 3 of 7 players
using a (2,3) QSS and the CK among all 7 players using a (6,7) CSS.  Any fewer than 6 players have no information about the CK, hence none about the secret.
A subset of 6 or more players, however, can recover the CK and will also contain at least 2 of the 3 players with quantum shares, hence can also recover the EQS
and decrypt to obtain the secret.  Thus the overall scheme is a perfect (6,7) QSS scheme, but requiring only 3 quantum shares.}\label{fig-hybrid}
\end{figure}

Assuming then that an {\it optimal} threshold QSS construction exists for the underlying $(k_q,n_q)$ scheme with a secret of size $d_s$
(as is always the case \cite{CGL} for prime $d_s\ge n_q$), we can implement a $(k,n)$ scheme with total dealer quantum communication, in qubits, of
\begin{equation}
Q_{HQSS}=(2n-2k+1)\log_2 (d_s).\label{eq-hqsscost}
\end{equation}
All of the perfect hybrid schemes described in \cite{NMI01} and later work by 
Singh and Srikanth \cite{SS05} are based on distributing the EQS using perfect QSS, which reduces the number
of quantum shares but not their individual sizes.  In the next section we demonstrate
that one may reduce both.

\section{Hybrid ramp QSS schemes}
The above scheme combines perfect QSS and CSS schemes to create a perfect hybrid QSS.
We note that, if one instead combines a  {\it ramp} QSS scheme with a perfect CSS scheme,
one can potentially also create a perfect hybrid ramp QSS (HRQSS) scheme, despite the presence
of intermediate sets with respect to the underlying RQSS scheme.  If the shares are distributed
in such a way that any subset of players with partial information on the EQS has no information on the CK,
then overall the subset has no information about the original secret and hence is forbidden.
In this case there are no longer any intermediate sets and the scheme is perfect.

To this end, the same general construction used in Lemma \ref{lem-ht}
can be used, but with a $(k_{qr},L_{qr},n_{qr})$ underlying RQSS instead of
a $(k_q, n_q)$ QSS.
By the same reasoning as (\ref{eq-kngap}) we have the requirement
\begin{equation}
n_{qr}-k_{qr}\ge n-k\label{eq-rampgap}
\end{equation}
Since clearly $n_{qr} \le n$, it follows from (\ref{eq-rampgap}) that
$k_{qr}\le k$, and hence any intermediate subset of the RQSS will be forbidden with respect to the CK
and hence the secret, as required.  There is, however, one important difference from the previous case:
Since gaining information about an unknown quantum state implies disturbance of that state \cite{FP96},
for any $k_{qr}$ players to have complete knowledge of the secret requires that the complementary
sets of $n_{qr}-k_{qr}$ players have no knowledge of the secret.  In a $(k_{qr},L_{qr},n_{qr})$ RQSS scheme
a subset of more than $k_{qr}-L_{qr}$ players has some knowledge of the secret, hence we require for such a scheme that
\begin{align}
n_{qr}-k_{qr}&\le k_{qr}-L_{qr},\textrm{ thus}\nonumber\\
n_{qr}&\le 2k_{qr}-L_{qr}.\label{eq-rampbound}
\end{align}
Hence we cannot have $n_{qr}=2k_{qr}-1$ (except in the trivial $L_{qr}=1$ case) and
simply replace the underlying boundary
QSS with an RQSS.  Instead we must use an RQSS which, we will see,
in general has more quantum shares than the QSS in our previous construction (though still
fewer than a standard $(k,n)$ QSS), but the quantum shares are smaller.
We find the following result:

\begin{thm}\label{thm-hrqss}
One may construct a perfect $(k,n)$ threshold scheme with secret size $\log_2 d_s$ qubits and $n_{qr}$ partly- or fully-quantum
shares (where $2n-2k+1\le n_{qr}\le n$)  with shares of size $\frac{\log_2 d_s}{1-2\frac{n-k}{n_{qr}}}$ qubits,
for a total dealer quantum communication cost, in qubits, of
\begin{equation}
Q_{HQRSS}=\frac{\log_2 d_s}{1-2\frac{n-k}{n_{qr}}},\label{eq-hqrsscost}
\end{equation}
provided that an optimal $(k_{qr},L_{qr},n_{qr})$ RQSS exists for a secret of size $d_s$,
where $k_{qr}=n_{qr}-(n-k)$ and $L_{qr}=2k_{qr}-n_{qr}$.
\end{thm}

\begin{proof}
An optimal $(k_{qr}, L_{qr},n_{qr})$ RQSS protocol has share size $\frac{n_{qr}}{L_{qr}}\log_2 d_s$ qubits,
(we can always, though not exclusively, construct such protocols \cite{OSIY05} for shares of prime dimension $\ge n_{qr}$).
For given $n_{qr}$, we wish to minimise
\begin{equation}
\frac{n_{qr}}{L_{qr}}\ge \frac{1}{2\frac{k_{qr}}{n_{qr}}-1}
\end{equation}
subject to (\ref{eq-rampbound}) and (\ref{eq-rampgap}).  We can do this by setting
\begin{align}
k_{qr}&=n_{qr}-(n-k)\\
L_{qr}&=2k_{qr}-n_{qr}.
\end{align}
We therefore, as in the HQSS case, encrypt the secret using a CK
of $2 \log_2 d_s$ bits and distribute the CK among the players
using a $(k,n)$ threshold CSS scheme, then distribute the corresponding
EQS among $n_{qr}$ of the players using an optimal $(k_{qr},L_{qr}, n_{qr})$ RQSS scheme
satisfying the above parameters.
\end{proof}

As seen from (\ref{eq-hqrsscost}), the dealer's quantum communication cost $Q_{HQRSS}$ decreases as $n_{qr}$ increases
and hence is at its lowest
when $n_{qr}=n$ i.e.\ when all $n$ shares have a small quantum element of size $(\log_2{d_s})/n_{qr}$, in which case
the total quantum communication cost is
\begin{equation}
Qmin_{HQRSS}=\frac{n\log_2 d_s}{2k-n}.\label{eq-minhqrsscost}
\end{equation}
The minimum
number of quantum shares (and hence the largest cost) occurs when
$n_{qr}=2n-2k+1$ and $L_{qr}=1$ i.e.\ the HQSS case of Lemma \ref{lem-ht}, with a cost given
by (\ref{eq-hqsscost}).  In general, then,
this scheme has a lower cost than the HQSS construction.

Hence we have a range of possible protocols; we can use fewer, larger quantum shares
(as would be suitable when e.g. many players can only process classical information)
or more, smaller quantum shares (suitable when e.g. more players can process
quantum information, but with small communication or storage capacities), with the
latter option giving the smallest overall quantum system distributed by the dealer.

\section{Bounds on quantum and classical share size}\label{sec-bound}
In our above constructions of perfect QSS schemes, involving both quantum
and classical elements to the player shares, we have only considered optimisation
with respect to the quantum communication cost.  It is straightforward, however,
to derive a joint bound on both quantum and classical share sizes,
by a simple generalisation of the proof of the quantum share size bound
given in \cite{Gott00} (since that bound was ultimately established by
referring to classical communication cost), which we give below, using much
of the same reasoning as \cite{Gott00}.
\begin{thm}
A perfect QSS scheme sharing an arbitrary quantum secret of dimension $d_s$ requires shares with quantum elements
of dimension $d_q$ and classical elements of dimension $d_c$ such that, for every important share (those which can affect
whether or not a subset is authorised), $2\log_2 d_q + \log_2 d_c \ge 2\log_2 d_s$.
\end{thm}
\begin{IEEEproof}
As shown in \cite{Gott00}, if an arbitrary $d$-dimensional quantum secret $\psi$ is encoded into a state of $n$ shares $\phi_{p_1p_2\ldots p_n}$,
such that a subset of shares $P$ can fully recover $\psi$ in the absence of any assistance from the
complementary subset $\overline{P}$, then the subset $P$ can also, without assistance, {\it alter} the
secret $\psi$ to a new arbitrary $d$-dimensional state $\psi'$, such that a recovery procedure that would originally
have recovered $\psi$ will now recover $\psi'$.

Consider now a perfect pure-state QSS (i.e.\ the combined state of all players for a given secret is pure),
for a $d_s$-dimensional secret $\psi_{d_s}$, and a forbidden subset of players  $F$ such
that $F$ can be made authorised by adding one additional player.
For a perfect scheme, it is clear that such a subset must always exist (and can be found for any given player
with an important share).  Additionally, in a pure-state QSS, it is known \cite{CGL}
that the complement of any unauthorised subset is authorised, so the complementary subset $\overline{F}$ of $F$
will be authorised.

The subset $\overline{F}$ can therefore replace $\psi_{d_s}$ with an arbitrary $d_s$-dimensional secret $\psi_{d_s'}$.  They can then
choose an appropriate player $p_a$ and pass that player's share (which we will take to be a quantum share of dimension $d_q$ and a
classical share of dimension $d_c$, abbreviated $(d_q,d_c)$) to the subset $F$, such that the resultant
subset $Fp_a$ will be authorised and able to recover the new secret $\psi_{d_s'}$.  Consequently, this construction provides
a mechanism to transfer an arbitrary quantum state $\psi_{d_s'}$ of dimension $d_s$ between parties by transferring a share of dimension $(d_q,d_c)$
(note though that this mechanism may require the parties to additionally share some pre-existing entanglement).
Since any perfect mixed-state QSS
can be represented as a perfect pure-state QSS with some shares discarded \cite{CGL}, mixed-state QSS schemes also imply the existence
of such a mechanism.

Consider now the case of two parties, Alice and Bob, who share a cat state $\sum_{i=0}^{d_s}\ket{i}_A\ket{i}_B$.
By the simple generalisation of superdense coding \cite{dense}, Alice may locally encode one of ${d_s}^2$ classical states
into this state, and thus communicate $2\log_2 {d_s}$ classical bits to Bob by transferring her $d_s$-dimensional half of
the state to him.  Hence if Alice and Bob possess appropriate states to make use of the QSS mechanism above,
Alice can transfer to Bob a state of dimension $(d_q,d_c)$ and thereby transfer a state of dimension $d_s$ and hence communicate $2\log_2 d_s$ classical bits.

However, it has been proven \cite{densebound} that to communicate $2\log_2 d$ classic bits using quantum states, even in the presence of pre-existing entanglement, requires the communication
of a quantum state of dimension $\ge d$.  Hence transferring a quantum state of dimension $d_q$ can communicate at most $2\log_2 d_q$ bits.
Since the entanglement-assisted classical capacity of a channel is additive (and the entanglement-assisted capacity of a perfect classical channel
is simply its classical capacity) \cite{BSST02}, transferring a state $d_q$ and $\log_2 d_c$ classical bits can transfer at most
$2\log_2 d_q+\log_2 d_c$ bits.  Hence communicating $2\log_2 d_s$ bits in total via this method (as we have shown can be achieved given the existence of a perfect QSS with this share size)
requires that
\begin{equation}
\log_2 d_c + 2\log_2 d_q \ge 2\log_2 d_s\label{eq-sharebound}
\end{equation}
and hence this bound must be satisfied by all important shares of any perfect QSS.
\end{IEEEproof}
We require a CK of $2\log_2 d_s$ bits to classically
encrypt a secret of size $d_s$, and a perfect CSS requires classical share sizes $\log_2 d_c$ to be at least as
large as the classical secret.  Hence (since we clearly require $d_q>0$ for a quantum secret), if using a hybrid construction in which the quantum secret is securely classically 
encrypted, we cannot saturate the bound (\ref{eq-sharebound}) using a perfect CSS.
We show in the next section that, however, we can saturate this bound, and produce a perfect
overall QSS, using a hybrid scheme in which both underlying QSS and CSS schemes are ramp schemes.

\section{Doubly-ramp optimal HRQSS}
Our construction for such a hybrid scheme will be much as before: encrypt the quantum
secret using a CK, then distribute the EQS and CK using a QSS and CSS respectively.
In general, however, if both underlying schemes are ramp schemes, certain
subsets of players will be intermediate subsets with respect to both the EQS and CK
i.e.\ they will have some partial information about both, and hence in general be able to construct
some partial information about the secret, making the scheme imperfect.

We find, though, that one can construct schemes which distribute the information
in such a way that such subsets, despite having
partial information on both EQS and CK, have no information about the original quantum secret.
Hence we can use both quantum and classical ramp schemes in a hybrid scheme,
reducing the amount of classical communication required by the dealer
and the size of the classical element of individual shares.

We will describe such schemes as optimal if the size of the largest quantum share element
is the smallest possible and all shares saturate the bound (\ref{eq-sharebound}) i.e. given the size of the quantum
elements, the classical elements are also the smallest possible.  (In our examples below, all of the quantum elements
are of equal size, likewise the classical elements).

\subsection{$(n,n)$ protocols}
A relatively simple example occurs for the $(n,n)$ access structure,
for which we find the following result:
\begin{thm}
One can construct an optimal perfect $(n,n)$ hybrid threshold QSS encoding
an $n$-qudit secret (where the qudits are of equal but arbitrary dimension).
\end{thm}
\begin{proof}
Separately encrypt each qudit $q_i$ using a 2-dit CK $k_i$ via
the standard classical encryption protocol (\ref{eq-encrypt}).
Distribute each encrypted qudit to a different player $i$,
and distribute the corresponding $K_i$ to the remaining $n-1$ players
using an $(n-1,n-1)$ threshold CSS\footnote{An $(n,n)$ CSS with a $d$-dimensional secret can
be constructed using $d$-dimensional shares for any $n$ and $d$, by randomly
choosing the values of any $n-1$ shares, and choosing the value of the
final share
such that summation of all the shares modulo $d$ is equal to the secret.}.
Thus $n$ players will have all qudits and the corresponding CK for each qudit, hence
be able to recover the secret.
Any fewer than $n$ players will, for any given $q_i$, lack either the encrypted
qudit and/or sufficient shares to have any information about the corresponding
$K_i$, and hence have no information about the secret.

The total quantum communication is simply the size of the secret and hence
can be no smaller, and since every quantum share is of equal size, this also
minimises the size of the largest quantum share.  The share sizes (each player receives 2 classical dits for $n-1$ of the $K_i$
for a total of $2(n-1)$ dits and 1 qudit per share, for an $n$-qudit secret)
saturate the bound \ref{eq-sharebound}, hence the
scheme is optimal.
\end{proof}
We note that in the above construction (for which the $(2,2)$ case with a 1-qubit secret is simply the teleportation
scheme of \cite{HBB99}), even though each {\it individual}
CSS for a given CK $K_i$ is perfect (as is the overall HRQSS), the overall CSS scheme for the 
complete CK (consisting of all $K_i$) is a $(k,L,n)=(n,2,n)$ ramp scheme.
Any set of $n-2$ players has no information about any $K_i$,
while $n-1$ players will have complete information about a single $K_i$
and none about any of the others i.e. they will have partial information about the complete CK.  Interestingly,
the overall scheme is optimal despite this ramp CSS (RCSS) scheme not being optimal,
in the sense that the classical share size of $2(n-1)$ dits is greater than
the secret size of $2n$ dits divided by $L=2$.

\subsection{$(n-1,n)$ protocols}
The $(n,n)$ access structure allows for a relatively simple construction which
does not readily extend to other access structures.  As discussed below, however,
we have found a more complex approach which can be used to generate various
cases of $(n,n-1)$ protocols.  The underlying reasoning does share a common idea
with the $(n,n)$ protocol: since, as discussed above, we can perform a 2-bit classical
 encryption of a qubit (which can also be thought of as applying one of the four Pauli operations
 $\{I,X,Y,Z\}$ to that qubit at random), then we can construct a classical key for an encrypted
multi-qubit state from the keys for the individual encrypted qubits.  If the qubits and information
about the key are distributed in an HRQSS such a way that a subset of players only receives classical
 information pertaining to the qubits they do not possess, then overall they have no information
on the quantum secret.

We first define the code space we will be be using for our underlying QSS scheme.
Let $\ket{s;t}$ with $\{s,t\}\in\{0,1\}$ be the 2-qubit eigenstates satisfying
\begin{align}
X\otimes X\ket{s;t}&=(-1)^s\ket{s;t}\textrm{, and}\\
Z\otimes Z\ket{s;t}&=(-1)^t\ket{s;t}.
\end{align}
(These are also the four Bell states: $\ket{0,0}=\ket{\Phi^+}$, $\ket{1,0}=\ket{\Phi^-}$, $\ket{0,1}=\ket{\Psi^+}$ and $\ket{1,1}=\ket{\Psi^-}$).  Define a state of $n=2m$ qubits
\begin{equation}
|{\bf s};{\bf t}\rangle\equiv |s_1;t_1\rangle \otimes |s_2;t_2\rangle\otimes\cdots\otimes |s_m;t_m\rangle
\end{equation}
where  ${\bf s}=(s_1,s_2,...,s_m)$ and ${\bf t}=(t_1,t_2,...,t_m)$ are two sequences of $m$ bits each.

Consider now the stabiliser group $G_s$ acting on $n=2m$ qubits, with two generators:
\begin{equation}
G_s=\langle X^{\otimes n},Z^{\otimes n}\rangle.
\end{equation}
Our code space $\mathcal{C}$, in which we encode our quantum secret, will be the $n$-qubit subspace stabilised by $G_s$.  It can easily be verified that $g|{\bf s};{\bf t}\rangle=|{\bf s};{\bf t}\rangle$ for all four elements $g\in G_s$ if and only if $\sum_{k=1}^{m}s_k=0$ and $\sum_{k=1}^{m} t_k=0$ (with all summations done mod 2 here and henceforth), thus we can write $\mathcal{C}$ as
\begin{equation}
\mathcal{C}={\rm span}\left\{|{\bf s};{\bf t}\rangle\;\Big|\;\sum_{k=1}^{m}s_k=0\;\;,\;\;\sum_{k=1}^{m} t_k=0\right\}\;.
\end{equation}

We see that this code can correct 1 erasure error: should a qubit go missing, the remaining players can replace it with
a new qubit in some arbitrary state, then collectively measure the values of the stabilisers
$X^{\otimes n}$ and $Z^{\otimes n}$ on the new set of $n$ qubits, projecting the state back into the code space
up to an error on the replacement qubit.  Any error (indicated by the measured stabiliser values) can be corrected
by applying a $Z$ and/or $X$ operation to the new qubit, thus recovering the original state.

We wish to classically encrypt our encoded secret within this code space, and so introduce further notation
to describe our classical encryption.  Denote the four Pauli operations by $B_{pq}$ ($p,q\in \{0,1\}$) as follows:
\begin{equation}
B_{00}\equiv I,\quad B_{01}\equiv X,\quad B_{10}\equiv Z \quad B_{11}\equiv ZX=-iY.
\end{equation}
Then the 2-qubit states $\ket{s,t}$ satisfy
\begin{align}
B_{pq}\otimes I |s;t\rangle&= (-1)^{sq}|s+p;t+q\rangle\label{eq-BI}\\
I\otimes B_{pq} |s;t\rangle&= (-1)^{p(q+t)}|s+p;t+q\rangle\label{eq-IB}.
\end{align}
Given two $n$-bit strings ${\bf p}=(p_1,p_2,...,p_n)$ and ${\bf q}=(q_1,q_2,...,q_n)$, we denote
\begin{equation}
B_{\bf pq}=B_{p_1q_1}\otimes B_{p_2q_2}\otimes\cdots\otimes B_{p_nq_n}.\label{eq-Bpq}
\end{equation} 
It follows from (\ref{eq-Bpq}), (\ref{eq-BI}) and (\ref{eq-IB}) that $B_{\bf pq}|\psi\rangle\in\mathcal{C}$ for all $|\psi\rangle\in\mathcal{C}$ if and only if $\sum_{k=1}^{n}p_k=\sum_{k=1}^{n}q_k=0$.

Consider now a quantum secret $|\psi\rangle\in\mathcal{C}$.  The dealer picks, with uniform probability,
two $n$-bit strings ${\bf p}=(p_1,p_2,...,p_n)$ and ${\bf q}=(q_1,q_2,...,q_n)$ (to remain within the code space, we require that
$\sum_{k=1}^{n}p_k=\sum_{k=1}^{n}q_k=0$, thus there are only $2^{2(n-1)}$ possibilities for the two strings)
and applies the operator $B_{\bf{pq}}$ to produce the EQS
\begin{equation}
\ket{\psi_{\bf{pq}}}=B_{\bf{pq}}\ket{\psi}.
\end{equation}
Note also, however, that the stabiliser generators $X^{\otimes n}$ and $Z^{\otimes n}$ (which leave the secret unchanged) can be expressed as operators $B_{\bf{p'q'}}$, and that, up to an irrelevant overall phase, products of $B$ operators are $B$ operators themselves i.e.
\begin{equation}
B_{\bf{pq}}B_{\bf{p'q'}}=\pm B_{(\bf{p}+\bf{p'})(\bf{q}+\bf{q'})}.
\end{equation}
Thus, up to an overall phase, there are four equivalent operators $B_{\bf{pq}}$ for a given transformation of the secret
(the original $B_{\bf{pq}}$ combined with one, both, or neither of the two stabiliser generators), for a total of $2^{2(n-1)-2}=2^{2(n-2)}$ distinct transformations of the secret.  Hence our CK consists of $2(n-2)$ bits of information, as is required for an $(n-2)$-qubit secret.  Similarly, defining $\mathbf{\bar{p}}=(p_1+1,p_2+1,\ldots, p_n+1)$ and similarly for $\bf \bar{q}$ we have that
\begin{equation}
|\psi_{\bf pq}\rangle=\pm|\psi_{\bf p\bar{q}}\rangle=\pm|\psi_{\bf \bar{p}q}\rangle=\pm|\psi_{\bf \bar{p}\bar{q}}\rangle\;.
\label{eq-bars}
\end{equation}

The dealer distributes the EQS by sending each player a qubit; since the EQS is within the code space,
any $n-1$ players will be able to recover it.  The overall quantum secret is of size $n-2$ qubits and each
player receives 1 qubit, thus the scheme is optimal with respect to the quantum share size (since player shares are the smallest
possible quantum system: 1 qubit).  From the share size bound for perfect QSS schemes it is implicitly clear that the
distribution of the EQS constitutes an RQSS (since the quantum shares are smaller than the EQS).

To have an optimal overall HRQSS scheme for this quantum secret we see from (\ref{eq-sharebound}) that the classical share size must be
$2((n-2)-1)=2(n-3)$ bits.  We therefore require an RCSS scheme distributing the the $2(n-2)$ classical bits of information
in the CK (strings $\bf p$ and $\bf q$, up to the equivalency relations discussed above which mean that the strings are each $n$
bits long
but carry $n-2$ bits of information about the EQS) to the players using shares of this size.
A convenient feature of our construction is that a
 scheme which works for one of $\bf p$ and $\bf q$ can also be used to distribute the other, so we will consider only a
scheme for distributing the information in $\bf p$, with each player receiving $n-3$ bits.

We now describe, in terms of the distribution of $\bf p$, the classical schemes we have found for specific player numbers $n$.
These are all of the form that subsets of $n-1$ players receive complete information (as required), those of $n-2$ players receive
one (useless) bit of
 information on the CK, and fewer than $n-2$ players receive no information.  For a subset of players to possess useful classical
information on the quantum secret, they must have some information corresponding to the qubits they also possess (e.g. a
subset of players 1,2 and 3 must have some information on the bits $p_1$, $p_2$ and/or $p_3$.  Note that since $\sum_k p_k=0$,
knowing the complement of such bits can provide such information (e.g. if there were 4 bits in total, knowing $p_4+p_1$ would also give $p_2+p_3$).

\subsubsection{4 players}
In this case the dealer generates a single additional random bit $z$, and distributes to each player $k$ the bit $p_{\sigma(k)}+z$ using the permutation
\begin{equation}
\sigma=\left(
\begin{array}{cccc}
1& 2& 3& 4\\
1 &4 &2 &3
\end{array} \right)\;.
\end{equation}
so e.g. player 1 receives $p_1+z$, player 2 receives $p_4+z$ etc.  Since they do not know $z$, no player alone has any information on the
secret $\bf p$.  Since $\sum_k p_k=0$ and $4z=0$ mod 2, 3 players know that the missing player's bit is equal to the sum of their bits,
and can thus determine all 4 bits (not knowing $z$, they could have either $\bf p$ or $\bf \bar{p}$ but these are equivalent as per
(\ref{eq-bars}).

Two players can find out $p_{\sigma(i)}+p_{\sigma(j)}=p_{5-\sigma(i)}+p_{5-\sigma(j)}$, but the construction of $\sigma$ is such that no pair $(\sigma(i),\sigma(j))$ is $(i,j)$ or $(5-i,5-j)$ for any $i\neq j$.  E.g. players 1 and 4 can obtain
$p_1+p_4=p_2+p_3$ but since they only possess qubits 1 and 2 this gives no overall information on the quantum secret.
Thus by applying this scheme to the distribution of $\bf p$ and $\bf q$ (with independently generated values of $z$ for each) we
produce an optimal $(3,4)$ HRQSS.

\subsubsection{6 players}
For 6 players the classical secret size (for $\bf p$ alone) is $n-2=4$ bits, and each player receives $n-3=3$ bits of information.
The dealer generates additional random variables $x_1,x_2,\ldots, x_6$ and $y_1,y_2,\ldots, y_6$,
requiring that
\begin{equation}
\sum_{i=1}^{6}x_i=0\label{eq-xsum}
\end{equation}
(but with no such constraint on the $y_i$).  Each player $k$ receives the following 3 bits: $x_k$, $y_k$ and
\begin{multline}
 s_k=x_{k+1}+x_{k+2}+y_{k+2}+y_{k+3}+p_{k+3}\\
+ (p_{k-1}\textrm{ (if }k\textrm{ even) or }p_{k-2}\textrm{ (if }k\textrm{ odd)}).
\end{multline}
The complete set of received bits is given in Table \ref{tab-bits}.
\begin{table}
\begin{center}
\begin{tabular}{|c|c|c|c|}
\hline
Player $1$&$x_1$&$y_1$&$x_2+x_3+y_3+y_4+p_4+p_5$\\
\hline
Player $2$&$x_2$&$y_2$&$x_3+x_4+y_4+y_5+p_5+p_1$\\
\hline
Player $3$&$x_3$&$y_3$&$x_4+x_5+y_5+y_6+p_6+p_1$\\
\hline
Player $4$&$x_4$&$y_4$&$x_5+x_6+y_6+y_1+p_1+p_3$\\
\hline
Player $5$&$x_5$&$y_5$&$x_6+x_1+y_1+y_2+p_2+p_3$\\
\hline
Player $6$&$x_6$&$y_6$&$x_1+x_2+y_2+y_3+p_3+p_5$\\
\hline
\end{tabular}
\end{center}
\caption{Listing of bits received by players in the 6-player RCSS scheme}\label{tab-bits}
\end{table}
  We now consider key examples, which suffice to show the scheme works in all cases.
\paragraph{Example 1: 4 adjacent players}
If players 1,2,3 and 4 collaborate, they know $x_1$ to $x_4$ and $y_1$ to $y_4$ and can subtract these (and the complement
$x_5+x_6=x_1+x_2+x_3+x_4$, using (\ref{eq-xsum}) from their $s_k$ bits, leaving knowledge of:
\begin{align*}
1:&p_4+p_5\\
2:&y_5+p_5+p_1\\
3:&x_5+y_5+y_6+p_6+p_1\\
4:&y_6+p_1+p_3.
\end{align*}
Variable $y_6$ appears only in bit 4, so cannot be eliminated, making bit 4 useless to the players.  The same occurs with $x_5$ (since $x_6$ does not appear) and bit 3.
With bits 3 and 4 excluded, $y_5$ appears only in bit 2, rendering it useless, and the only information the players have is the value of
$p_4+p_5=p_1+p_2+p_3+p_6$, which doesn't help them
since they don't have qubit 5 or qubit 6.

If instead we have players 2 to 5 (we require a second example due to the odd/even dependence in the distributed bits),
they know the corresponding $x$ and $y$ values and hence from their $s_k$ bits they know:
\begin{align*}
2:&p_5+p_1\\
3:&y_6+p_6+p_1\\
4:&x_6+y_6+y_1+p_1+p_3\\
5:&y_1+p_2+p_3.
\end{align*}
Bit 4 is useless due to $x_6$, making bits 3 and 5 also useless due to $y_6 $ and $y_1$ respectively
and leaving the players with only $p_5+p_1=p_2+p_3+p_4+p_6$, but lacking qubits 1 or 6 means the players have no knowledge of the secret.
By cyclic permutation these examples cover all cases of 4 adjacent players.

\paragraph{Example 2: 4 players, 3 adjacent}
If players 1,2,3 and 5 collaborate, by analogous reasoning to the above they know:
\begin{align*}
1:&y_4+p_4+p_5\\
2:&x_4+y_4+p_5+p_1\\
3:&x_4+y_6+p_6+p_1\\
5:&x_6+p_2+p_3
\end{align*}
Bit 3 is useless due to $y_6$.  The players can sum bits 2 and 5 and eliminate $x_4+x_6=x_1+x_2+x_3+x_5$ to get $y_4+p_1+p_2+p_3+p_5$, then sum with bit 1 to get $p_1+p_2+p_3+p_4=p_5+p_6$ which is
their only bit of information, and useless without qubits 4 or 6.

If players 2,3,4 and 6 collaborate they know from their additional bits:
\begin{align*}
2:&y_5+p_5+p_1\\
3:&x_5+y_5+p_6+p_1\\
4:&x_5+y_1+p_1+p_3\\
6:&x_1+p_3+p_5\\
\end{align*}
Bit 4 is useless due to $y_1$. Eliminating $y_5$ between bits 1 and 2 gives $x_5+p_5+p_6$, then eliminating $x_5+x_1$ with
bit 6 gives $p_1+p_3=p_2+p_2+p_5+p_6$, useless without qubits 3 or 5.
By cyclic permutation these examples cover all cases of 4 players with 3 adjacent.
\paragraph{Example 3: 4 players, 2 adjacent}
If players 1,2,4 and 5 collaborate they have
\begin{align*}
1:&x_3+y_3+p_4+p_5\\
2:&x_3+p_5+p_1\\
4:&x_6+y_6+p_1+p_3\\
5:&x_6+p_2+p_3\\
\end{align*}
Bits 1 and 4 are useless due to $y_3$ and $y_6$.  Eliminating $x_3+x_6$ between bits 2 and 4 gives $p_1+p_2+p_3+p_5=p_4+p_6$, useless without qubits 3 or 6.

If players 2,3, 5 and 6 collaborate they have
\begin{align*}
2:&x_4+y_4+p_5+p_1\\
3:&x_4+p_6+p_1\\
4:&y_1+p_1+p_3\\
6:&x_1+p_3+p_5\\
\end{align*}
Bits 2 and 4 are useless due to $y_4$ and $y_1$. Eliminating $x_4+x_1$ between bits 3 and 6 gives $p_1+p_3+p_5+p_6=p_2+p_4$, useless without qubits 1 or 4.  Cyclic permutation of the above examples covers all remaining cases, thus we have shown that no set of 4 players has any information on the quantum secret.

We still need to show that 5 players can recover all 4 bits of the classical secret, for which it is sufficient
to show that they obtain 4 independent bits of information about the string $\bf p$.  Suppose players 1 to 5
 collaborate.  Using (\ref{eq-xsum}) they know all values of $x_k$ and all $y_k$ except $y_6$.  They therefore have access to the
 bits $p_4+p_5$, $p_5+p_1$, $y_6+p_6+p_1$, $y_6+p_1+p_3$, $p_2+p_3$ and $p_3+p_5$.  Eliminating $y_6$ between
 the two bits containing it gives them access to the 4 independent bits $p_4+p_5$, $p_5+p_1$, $p_3+p_6$ and $p_2+p_3$.
If players 2 to 6 collaborate, they have access to $p_4+p_5$, $p_5+p_1$, $p_6+p_1$, $y_1+p_1+p_3$, $y_1+p_2+p_3$ and $p_3+p_5$,
and can obtain 4 independent bits $p_4+p_5$, $p_5+p_1$, $p_6+p_1$ and $p_1+p_2$.  Other cases are covered through cyclic permutation.  Hence all subsets of 5 players can recover the secret.

We have therefore demonstrated optimal $(3,4)$ and $(5,6)$ HRQSS schemes.  We conjecture that one can construct such $(n-1,n)$
schemes for $2n$ qubits for any value of $n$, however we note that the relative advantage in doing so decreases as $n$ gets large:
our ramp classical share size of $2(n-3)$ classical bits, while optimal for the HRQSS, is only 2 bits less than the lower bound for a perfect CSS scheme.  Thus the $(3,4)$ case is arguably the most important example, since the classical shares are half of what
 an optimal perfect CSS scheme would require in this case.

More significantly, this method of construction for HRQSS schemes holds great promise for adapting other QECCs from the wide
class of stabiliser codes.  By exploiting the structure of the underlying code and formulating the CSS element
in terms of Pauli transformations on individual qubits, we reduced the problem to a purely classical one of constructing a suitable
matching RCSS scheme for the underlying QSS scheme (we note that we are not aware of our particular RCSS schemes having appeared previously in the literature), for which
we can use two separate, identical schemes for the strings $\bf p$ and $\bf q$, further simplifying the problem.
The properties of the HRQSS could be established purely classically, without having to generate
density matrices for the combined classical and quantum information for the various player subsets.  Adopting this approach with
other stabiliser codes should make finding additional HRQSS schemes much more straightforward.

\section{Conclusion}
We have shown that one can produce perfect threshold QSS schemes with both quantum share sizes
and total dealer quantum communication below what has previously been shown, for any allowable non-boundary threshold $(k,n)$
access structure (i.e.\ any scheme where $n<2k-1$).  We have further shown $(n,n)$ and (for specific values of $n$) $(n-1,n)$ access structures
for which one can optimise both the quantum and classical share size i.e.\ our protocol gives the smallest
possible quantum shares and the smallest possible classical shares given the quantum share size.  These protocols
allow one to find an efficient scheme for situations where some variable number of players can process quantum
information, and one's priorities may vary from minimising the total quantum communication to minimising
the total number of quantum shares.

These results suggest various directions for further work.  Firstly, our optimal doubly-ramp
schemes are currently only known for specific access structures; it would be very desirable to generalise
this, and to know when and how an appropriate RCSS scheme can be constructed for a given RQSS in order to produce
a perfect overall scheme.  As discussed above, the approach used appears very promising for adapting to other stabiliser codes.
We note that the RQSS schemes used do not fall under
the general construction of \cite{OSIY05}; finding matching RCSS schemes for this construction
would produce a wide range of optimal hybrid schemes.

We further note that our doubly-ramp schemes can be seen to be optimal in the sense of having the smallest possible quantum
shares only due to special circumstances: for the $(n,n)$ schemes our total quantum communication is the size of the secret
and for the $(n-1,n)$ schemes the quantum shares have the smallest non-trivial dimension of 2.  In general, however,
we do not know what the smallest possible quantum shares are for a given perfect hybrid scheme, or indeed how large
the overall quantum system communicated by the dealer must be, beyond the obvious lower bound
of the secret size $d_s$; finding a tight bound would be an important step to identifying optimal schemes.

We have only addressed the application of hybrid ramp schemes to the construction of threshold
schemes, rather than more general access structures.  Some results on non-ramp hybridisation of such
schemes were found in \cite{NMI01} and \cite{SS05} and it seems likely that hybrid ramp schemes
could reduce share sizes for these access structures as well.  In general, though, little is currently
known about non-threshold QRSS schemes.

Fianlly we note, and thank an anonymous referee for pointing out,
that for HRQSS schemes to be useful in a practical insecure setting, without access to secure quantum channels
between individual players and players and dealer, will require additional features to provide
security against eavesdropping in secret distribution and/or recovery.  This is another
important direction for future work.

\section*{Acknowledgment}
We thank Ran Hee Choi, Aram Harrow and Barry Sanders for helpful discussions.

\bibliographystyle{IEEEtran}
\bibliography{hybrid-ieee}

\end{document}